\documentclass{amsart}

\usepackage{amssymb,amsfonts,amsthm}
\usepackage[latin1]{inputenc}
\usepackage{mathrsfs}
\usepackage{graphicx}
\usepackage{lmodern}
\usepackage{enumerate}
\usepackage{hyperref}
\usepackage{bbm}
\usepackage{tabularx}
\usepackage{amsaddr}

\newtheorem{theorem}{Theorem}

\newtheorem{definition}[theorem]{Definition}

\newtheorem{lemma}[theorem]{Lemma}
\newtheorem{proposition}[theorem]{Proposition}

\usepackage{etoolbox}
\makeatletter
\patchcmd\maketitle
{\uppercasenonmath\shorttitle}
{}
{}{}
\patchcmd\maketitle
{\@nx\MakeUppercase{\the\toks@}}
{\the\toks@}
{}
{}{}
\patchcmd\@settitle
{\uppercasenonmath\@title}
{}
{}{}
\patchcmd\@setauthors
{\MakeUppercase{\authors}}
{\authors}
{}{}
\makeatother

\usepackage{color}
\usepackage{url}

\usepackage[T1]{fontenc}
\usepackage{geometry}
\title[$\alpha$-Hypergeometric Uncertain Volatility Models]{$\alpha$-Hypergeometric Uncertain Volatility Models and their Connection to 2BSDEs}

\begin{document}
	
	\author[Zaineb Mezdoud, Carsten Hartmann, Mohamed Riad Remita, Omar Kebiri] { Zaineb Mezdoud$^{\dagger,1a}$, Carsten Hartmann$^{\ddagger,1b}$, Mohamed Riad Remita$^{\dagger,1c}$, Omar Kebiri$^{\ddagger,1d}$}
	\address{$^{\dagger}$ Laboratory of probabilities and statistics, University of Badji-Mokhtar Annaba, Algeria}
	\address{$^{\ddagger}$ Institute of Mathematics,Brandenburgische Technische Universit\"at Cottbus-Senftenberg, Germany}
	
	\footnote{
		E-mails: $^A$\href{mailto:e-mail address}{zaineb.mezdoud@univ-annaba.org}, \href{mailto:e-mail address}{zaineb.mezdoud@gmail.com},\;  
		$^B$\href{mailto:e-mail address}{carsten.hartmann@b-tu.de},\; $^C$\href{mailto:e-mail address}{r$\_$remita@yahoo.fr},\; $^D$\href {mailto:e-mail address}{omar.kebiri@b-tu.de}}
	\begin{abstract}
		In this article we propose a $\alpha$-hypergeometric model with uncertain volatility (UV) where we derive a worst-case scenario for option pricing. The approach is based on the connexion between a certain class of nonlinear partial differential equations of HJB-type (G-HJB equations), that govern the nonlinear expectation of the UV model and that provide an alternative to the difficult model calibration problem of UV models, and second-order backward stochastic differential equations (2BSDEs). Using asymptotic analysis for the G-HJB equation and the equivalent 2BSDE representation,  we derive a limit model that provides an accurate description of the worst-case price scenario in cases when the bounds of the UV model are slowly varying. The analytical results are tested by numerical simulations using a deep learning based approximation of the underlying 2BSDE.
	\end{abstract}
	
	
	
	\maketitle
	
	\section{Introduction}
	
	The classical option pricing problem based on the seminal work by Black and Scholes \cite{black1973pricing} assumes that the volatility of the underlying asset is constant over time. While the Black-Scholes model is still considered an important paradigm for option pricing, there is plenty of empirical evidences that the assumption of constant volatility is not adequate. In order to come up with more realistic models, various strategies have been proposed to treat the volatility of asset prices as a stochastic process \cite{hull1987pricing}. One of the most famous representative of the large class of stochastic volatility models is the Heston model \cite{heston1993closed} that has become the basis of many other models, such as jump diffusion models \cite{lee2010detecting}, $\alpha$-hypergeometric models  \cite{da2016alpha}, or various forms of uncertain volatility models (UVM) such as  \cite{avellaneda1995pricing,fouque2014approximation,guyon2013nonlinear}, all of which  can be considered as extensions of the Black-Scholes model and which share many features with the model considered in this article.
	
	One of the common feature of all  stochastic volatility models is that the volatility process can only be indirectly observed through the asset price, which poses specific challenges for the parameter estimation (or: \emph{calibration}) of these models. Standard approaches are based on maximum likelihood estimation using (filtered) time series data \cite{sahalia2007mle,hurn2015mle} or fitting of the implied volatility surface \cite{gatheral2006smile,gilli2021smile}. In the Heston model, the price hits zero in finite time unless the Feller condition is imposed. As a consequence, the underlying optimisation problems are typically endowed with constraints, which pose additional problems in model calibration.

	Here we use an alternative approach, in which we consider the unknown diffusion coefficient of the stochastic volatility model a bounded random variable. Specifically, we  focus on the UVM developed by \cite{avellaneda1995pricing} and consider an  $\alpha$-hypergeometric stochastic volatility model of the form
	\begin{subequations}\label{hyperG}
		\begin{align}\label{hyperGx}
			dX_{t} & =rX_{t}dt+X_{t}q e^{V_{t}}dW_{t}^{1}\\\label{hyperGv}
			dV_{t} & = (a-be^{\alpha V_{t}})dt+\sigma dW_{t}^{2},
		\end{align}
	\end{subequations}
	where $W^{1}_{t}$ and $W^{2}_{t}$ are correlated Brownian motions, with $d(W_{t}^{1}, W_{t}^{2}) = \rho dt$  for some $|\rho|\leq 1$, and $b,\alpha,\sigma > 0$ and $a\in \mathbb{R}$ are constants; the parameter $q$ is unknown; the only information available is that $q\in [\sigma_{\min}, \sigma_{\max}]$ for some  $\sigma_{\max},\sigma_{\min}\in\mathbb{R^*_+}$.
	This implies that the volatility $\beta_t$ of the risky asset under the risk-neutral measure $\mathbb{Q}$,
	\begin{equation}
		dX_{t}=rX_{t}dt+X_{t}\beta_{t}dW_{t}^{1}\,,
	\end{equation}
	where $r\in\mathbb{R}$ is the risk-free interest rate, is stochastic with $\underline{\sigma}_t\le \beta_t\le \overline{\sigma}_t$, with
	\[
	\underline{\sigma}_{t}=\sigma_{\min} F(V_{t})\,,\quad \overline{\sigma}_{t}:=\sigma_{\max} F(V_{t})\quad \textrm{for }\; 0\leq t\leq T\,.
	\]
	Here $F>0$ is a differentiable increasing function that we choose to be $F(v)=e^{v}$.
	
	Our aim is to derive worst-case pricing scenarios for the seller in the spirit of  the work \cite{da2016alpha}, without needing to calibrate the model exactly. 
	To this end, we rescale time in the volatility equation in (\ref{hyperG}) according to $t\mapsto \delta t$,  which yields
	\begin{subequations}\label{hyperGdelta}
		\begin{align}\label{deldX}
			dX_{t} & =rX_{t}dt+X_{t}q e^{V_{t}}dW_{t}^{1}\\\label{deldV}
			dV_{t} & = \delta (a-be^{\alpha V_{t}})dt+\sqrt{\delta}\sigma dW_{t}^{2}\,
		\end{align}
	\end{subequations}
	and allows us to smoothly interpolate between an UVM and a fixed volatility model (cf.~\cite{fouque2014approximation}). The parameter $\delta > 0$ symbolizes the reciprocal of the time-scale of the process $V $, and thus the standard UVM can be formally obtained by sending $\delta\to 0$, in which case $V_{t}=v$ and
	\begin{equation}\label{0dX}
		dX^{0}_{t}=rX^{0}_{t}dt+qX^{0}_{t}e^{v}dW_{t}^{1}.
	\end{equation}
	Varying $\delta$ sheds some light on the importance of the stochastic volatility equation for the worst-case scenario: when the variation of the volatility is slow, the market price of the asset is not very volatile, so this price remains stable; in the opposite case, it may becomes too volatile and therefore more risky.
	
	The rest of the paper is organised as follows: In Section \ref{sec:worstcase} we formulate the worst-case price scenario and the corresponding fully nonlinear partial differential equation of G-Hamilton-Jacobi-Bellman type (G-HJB equation), and we derive some basic properties such as moment bounds and the convergence of the worst-case price scenario as $\delta\to 0$; the Section also includes some technical results such as convergence of the second derivatives (Greeks).
	%
	In Section \ref{sec:2BSDE}, we consider the formulation of the fully nonlinear PDE for the nonlinear expectation of the price process and derive a uniform corrector result for the limit $\delta\to 0$ that complements the analysis of Section \ref{sec:worstcase}. We moreover formulate a concrete model that is solved numerically in Section \ref{sec:numerics}, using the deep learning method by Beck et al. \cite{beck2019machine} and exploiting the link between fully nonlinear G-HJB equations and 2BSDE. The main finding are summarised in Section \ref{sec:conclusion}.

	\section{Worst-case scenario price}\label{sec:worstcase}

	Let $\Theta=[\sigma_{\min},\sigma_{\max}]$. For any $\delta>0$, the worst-case scenario price at time $t < T$ is defined as
	\begin{equation}\label{Pdel}
		P^{\delta}:=P^{\delta}(t;x,v) = \exp(-r(T-t))\sup_{ q\in\Theta} E_{(t;x,v)}[h(X^{\delta}_{T})].
	\end{equation}
	If $\delta=0$, we define 
	\begin{equation}\label{P0}
		P^{0}:=P^{0}(t;x,v) = \exp(-r(T-t))\sup_{ q\in\Theta} E_{(t;x,v)}[h(X^{0}_{T})].
	\end{equation}
	Where $E_{(t;x,v)}[\cdot] $ is the conditional expectation given $\mathcal{F}_{t}$  with $X^{\delta}_{t}=x$ and $V_{t} = v$. \\
	

	\subsection{Moment bounds} 
	Instead of confining ourselves to perturbations of Black-Scholes prices as in \cite{fouque2011multiscale}, we will work with general terminal payoff (neither convex, nor concave ) as in \cite{fouque2018uncertain}. In this case the Hessian of the resulting option prices is indefinite and we have to impose additional regularity conditions on the payoff function $h$ to do some asymptotic analysis. Specifically, we suppose that the terminal payoff  $h$ is $C^4$ and gradient Lipschitz, and we impose the following polynomial growth conditions on the first four derivatives of $h$:
	\begin{equation}
		\left\{
		\begin{array}{rl}
			
			|h'(x)|\leq K_{1}, \qquad\qquad\quad\\
			|h''(x)|\leq K_{2}(1+|x|^{m}),\\
			
			|h'''(x)|\leq K_{3}(1+|x|^{n}), & \qquad\qquad\qquad (K_{i} \mbox{ for } i \in\{1,2,3,4\})  m, n \mbox{ and } l \in \mathbb{N},\\
			
			|h^{(4)}(x)|\leq K_{4}(1+|x|^{l}).
			
		\end{array}
		\right.
	\end{equation}

	Before we come to the convergence of $P^{\delta}$ as $\delta\to 0$, the next two propositions show that the processes  $X_{t}$ and $V_{t}$ have uniformly bounded moments of any order.
	
	\begin{proposition}\label{inq1} Let $0 \leq \delta \leq 1$, for $t \leq T$. The process $V_{t}$ has uniformly bounded  moments of any order
		\[
		\mathbb{E}_{(t,x,v)}\left[\int^{T}_{t} |V_{s}|^{k}ds\right] \leq \mathbb{E}_{(0,v)}\left[\int^{T}_{0} |V_{s}|^{k}ds\right] \leq C_{k}(T,v),
		\]
		where $C_{k}(T,v)$ independent of $\delta$.
	\end{proposition}
	
	\begin{proof} See Lemma $4.9$ in \cite{fouque2011multiscale}.
	\end{proof}
	
	\begin{lemma}
		For $\eta\in\mathbb{R}$ independent of $0 \leq \delta \leq \delta_0$, for some sufficiently small $\delta_0>0$, and $t \leq T$, the moment generating function of the integrated $\alpha$-hypergeometric process
		\[
		M^{\delta}_{v}(\eta):=\mathbb{E}_{(t,v)}[e^{\eta\int^{t}_{0} V_{s}ds}],\qquad\qquad \mbox{ for } \eta\in\mathbb{R},
		\] 
		is uniformly bounded, that is $|M^{\delta}_{v}(\eta)|\leq N(T,v,\eta) < \infty$, where $N(T,v,\eta)$ is independent of $t$.
	\end{lemma}
	\begin{proof}
		Following the reasoning of \cite[Sec.~5]{lepage2006continuous}, we have an explicit form of the moment generating function of the integrated $\alpha$-hypergeometric process: 
		$$
		M^{\delta}_{v}(\eta) = \Psi(\eta,t)e^{-v\Xi(\eta,t)},
		$$
		where
		\begin{eqnarray*}
			\Psi(\eta,t)&=&\left(\frac{\bar{b}e^{\delta \frac{t}{2}}}{\bar{b}\cosh(\bar{b}\frac{t}{2})+\delta \sinh(\bar{b}\frac{t}{2})}\right)^{\frac{2}{\sigma^{2}}},\\
			\Xi(\eta,t)&=&\left(\frac{2\eta\sinh(\bar{b}\frac{t}{2})}{\bar{b}\cosh(\bar{b}\frac{t}{2})+\delta\sinh(\bar{b}\frac{t}{2})}\right)^{\frac{2}{\sigma^{2}}},
		\end{eqnarray*}
		and
		$$\bar{b}=\sqrt{\hat{b}^{2}-2\eta\hat{\sigma}^{2}}=\sqrt{\delta^{2}-2\eta\delta\sigma^{2}}.$$
		In the following, we are going to show that $|M^{\delta}_{v}(\eta)| \leq N(T,v,\eta) < \infty$ , where $N(T,v,\eta)$ is independent of $\delta$ and $t$. To this end, we distinguish two cases: 
		\begin{itemize}
			\item[$\bullet$]  If $\delta^{2}-2\eta\delta\sigma^{2} \geq 0$, we have $\bar{b} \geq 0$ and
			$$
			\begin{array}{lllcc}
				\Psi(\eta,t) \leq \left( \frac{\bar{b}e^{\delta\frac{t}{2}}}{\bar{b}\cosh(\bar{b}\frac{t}{2})}\right)^{\frac{2}{\sigma^{2}}}, \qquad &\qquad   \delta\sinh(\bar{b}\frac{t}{2})\geq 0,\\
				\qquad\quad\leq \left(e^{\delta\frac{t}{2}} \right)^{\frac{2}{\sigma^{2}}}, \qquad & \qquad \cosh(\bar{b}\frac{t}{2})\geq 1, \\
				\qquad\quad\leq \left(e^{\frac{T}{2}} \right)^{\frac{2}{\sigma^{2}}}.
			\end{array}
			$$
			Since $\Xi(\eta,t) \geq 0$, we have $e^{-v\Xi(\eta,t)} \leq1$. Therefore
			$$
			M^{\delta}_{v}(\eta) = \Psi(\eta,t)e^{-v\Xi(\eta,t)}\leq\left(e^{\frac{T}{2}} \right)^{\frac{2}{\sigma^{2}}}.
			$$
			\item[$\bullet$]  If $\delta^{2}-2\eta\delta\sigma^{2} < 0$, let $\vartheta =\sqrt{2\eta\delta\sigma^{2}-\delta^{2}}$ which is positive. Then
			$$M^{\delta}_{v}(\eta) = \psi(\eta,t)e^{-v\Xi(\eta,t)},$$
			\begin{eqnarray*}
				&=& \left(\frac{i\vartheta e^{\delta \frac{t}{2}}}{i\vartheta\cosh(i\vartheta\frac{t}{2}) +\delta \sinh(i\vartheta\frac{t}{2})}\right)^{\frac{2}{\sigma^{2}}} e^{-v\left(\frac{2\eta\sinh(i\vartheta\frac{t}{2})}{i\vartheta\cosh(i\vartheta\frac{t}{2}) +\delta \sinh(i\vartheta\frac{t}{2})}\right)^{\frac{2}{\sigma^{2}}}},\\
				&=&\left(\frac{i\vartheta e^{\delta \frac{t}{2}}}{i\vartheta\cos(\vartheta\frac{t}{2}) +i\delta \sin(\vartheta\frac{t}{2})}\right)^{\frac{2}{\sigma^{2}}} e^{-v\left(\frac{2i\eta\sin(\vartheta\frac{t}{2})}{i\vartheta\cos(\vartheta\frac{t}{2}) +i\delta \sin(\vartheta\frac{t}{2})}\right)^{\frac{2}{\sigma^{2}}}},\\
				&=&\left(\frac{\vartheta e^{\delta \frac{t}{2}}}{\vartheta\cos(\vartheta\frac{t}{2}) +\delta \sin(\vartheta\frac{t}{2})}\right)^{\frac{2}{\sigma^{2}}} e^{-v\left(\frac{2\eta\sin(\vartheta\frac{t}{2})}{\vartheta\cos(\vartheta\frac{t}{2}) +\delta \sin(\vartheta\frac{t}{2})}\right)^{\frac{2}{\sigma^{2}}}}.
			\end{eqnarray*}
			Thus, for sufficiently small $\vartheta$, since
			$$
			\left(\frac{2\eta\sin(\vartheta\frac{t}{2} )}{\vartheta\cos(\vartheta\frac{t}{2})+\delta\sin(\vartheta\frac{t}{2})}\right)^\frac{2}{\sigma^{2}} \geq0,
			$$
			we have
			\begin{eqnarray*}
				M^{\delta}_{v}(\eta) &\leq& \left(\frac{\vartheta e^{\delta\frac{t}{2}}}{\vartheta \cos(\vartheta\frac{t}{2}) +\delta\sin(\vartheta\frac{t}{2})}\right)^\frac{2}{\sigma^{2}},\qquad\qquad\qquad\qquad\\
				&=&\left(\frac{\vartheta e^{\delta\frac{t}{2}}}{\vartheta(1 +\mathcal{O}(\vartheta^{2}t^{2}))+\delta(\frac{\vartheta t}{2} +\mathcal{O}(\vartheta^{3}t^{3}))}\right)^\frac{2}{\sigma^{2}},\\
				&=&\left(\frac{e^{\delta\frac{t}{2}}}{1 + \frac{\delta t}{2}+\mathcal{O}(\vartheta^{2}t^{2})}\right)^\frac{2}{\sigma^{2}}.\qquad\qquad\qquad\quad.
			\end{eqnarray*}
			As a consequence, there exists $\vartheta_{0}$ independent of  $t$, such that for $\vartheta< \vartheta_{0}$, \\  $M^{\delta}_{v}(\eta)\leq\left(\frac{e^{\frac{T}{2}}}{1+\frac{T}{2}}\right)^{\frac{2}{\sigma^{2}}}$.
		\end{itemize}
		This concludes the proof. 
	\end{proof}
	\begin{proposition}\label{prop2}
		Let $\delta\ge 0$ be sufficiently small and for $t \leq T$. Then the process $X_{t}$ has uniformly bounded  moments of arbitrary order. 
	\end{proposition}
	
	\begin{proof}
		
		Let $X_{t},V_t$ satisfy (\ref{hyperGdelta}), with $q_{t}\in[\sigma_{\min},\sigma_{\max}]$. Then, for each finite $n\in\mathbb{N}$, 
		\begin{eqnarray*}
			X^{n}_{t}&=&x^{n}\exp\left(nrt-\frac{n}{2}\int^{t}_{0}(q_{s}e^{V_{s}})^{2}ds + n\int^{t}_{0}q_{s}e^{V_{s}}dW_{s}^{1}\right)\\
			&=&x^{n} \exp\left(nrt + \frac{n^{2}-n}{2}\int_{0}^{t}(q_{s}e^{V_{s}})^{2}ds\right)
			\exp\left(\frac{-n^{2}}{{2}}\int^{t}_{0}(q_{s}e^{V_{s}})^{2}ds + n\int^{t}_{0}q_{s}e^{V_{s}}dW_{s}^{1}\right)\\
			&\leq& x^{n} \exp\left(nrt + \frac{n^{2}-n}{2}\int_{0}^{t} \sigma_{\max}^{2}e^{2V_{s}}ds\right)\Lambda_{t}\,,
		\end{eqnarray*}
		where in the last step we assume Novikov's condition which implies that 
		$$\Lambda_{t}=\exp\left(\frac{-n^{2}}{{2}}\int^{t}_{0}(q_{s}e^{V_{s}})^{2}ds + n\int^{t}_{0}q_{s}e^{V_{s}}dW_{s}^{1}\right)$$
		is a martingale.\\
		
		\noindent
		Using Proposition~\ref{inq1}, we find  
		\begin{eqnarray*}
			\mathbb{E}_{(0,x,v)}\left[ \exp\left(\frac{1}{2}\int_{0}^{t}(nqe^{V_{s}})^{2}ds\right) \right]&\leq& \mathbb{E}_{(0,x,v)}\left[ \exp\left(\frac{n^{2} u^{2}}{2}\int_{0}^{t}e^{2V_{s}}ds\right) \right]\\
			&=&\mathbb{E}_{(0,x,v)}\left[ \exp\left(\frac{n^{2} \sigma_{\max}^{2}}{2}\int_{0}^{t}\left(1+2V_{s}+\mathcal{O}((2V_{s})^{2})\right) ds\right) \right]\\
			&=&\mathbb{E}_{(0,x,v)}\left[ \exp\left(\frac{n^{2} \sigma_{\max}^{2}}{2}\left[\int_{0}^{t}ds+2\int_{0}^{t}V_{s}ds+\int_{0}^{t}\mathcal{O}((2V_{s})^{2}) ds\right]\right) \right]\\
			&=&\mathbb{E}_{(0,x,v)}\left[ \exp\left(\frac{n^{2} \sigma_{\max}^{2}}{2}(t+C)+2\frac{n^{2} \sigma_{\max}^{2}}{2}\int_{0}^{t}V_{s}ds\right) \right]\\
			&=&\mathbb{E}_{(0,x,v)}\left[ \exp\left(\frac{n^{2} \sigma_{\max}^{2}}{2}(t+C)\right).\exp\left((n^{2} \sigma_{\max}^{2})\int_{0}^{t}V_{s}ds\right) \right]\\
			&=&\exp\left(\frac{n^{2} \sigma_{\max}^{2}}{2}(t+C)\right) \mathbb{E}_{(0,x,v)}\left[\exp\left((n^{2} \sigma_{\max}^{2})\int_{0}^{t}V_{s}ds\right) \right]\\
			&=&\exp\left(\frac{n^{2} \sigma_{\max}^{2}}{2}(t+C)\right) M^{\delta}_{v}(n^{2} \sigma_{\max}^{2})\\
			&<& \infty.
		\end{eqnarray*}
		Hence,
		\begin{eqnarray*}
			\mathbb{E}_{(0,x,v)}[X^{n}_{t}]&\leq& x^{n}\exp(nrt)\mathbb{E}_{(0,x,v)}\left[\exp\left(\frac{(n^{2}-n) \sigma_{\max}^{2}}{2}\int_{0}^{t}e^{2V_{s}}ds\right) \right],\\
			& = & x^{n}\exp(nrt)\exp\left(\frac{(n^{2}-n) \sigma_{\max}^{2}}{2}(t+C)\right) M^{\delta}_{v}\left((n^{2}-n) \sigma_{\max}^{2}\right),\\
			&\leq& x^{n}\exp(nrT)\exp\left(\frac{(n^{2}-n) \sigma_{\max}^{2}}{2}(T+C)\right) N\left(T,v,(n^{2}-n) \sigma_{\max}^{2}\right):=L,
		\end{eqnarray*}
		where the upper bound $L$ is independent of $\delta$ and $t$.
		
		Therefore,
		$$\mathbb{E}_{(t,x,v)}\left[\int^{T}_{t} |X_{s}|^{k}ds\right] \leq\mathbb{E}_{(0,x,v)}\left[\int^{T}_{0} |X_{s}|^{k}ds\right] \leq N_{k}(T,x,v),$$
		where $N_{k}(T,x,v)$ may depend on $(k,T,x,v)$ but not on $\delta$.
	\end{proof}

	\subsection{Convergence of the payoff}
	
	As a consequence of the previous results, we have the following convergence result for the asset process.

	\begin{proposition}\label{prop3}
		Assume there exists $C_{0}>0$, independent of $\delta$, such that $X^{\delta},$ $X^{0}$ being the solution of the SDEs (\ref{deldX}) and (\ref{0dX}) satisfy  
		$$E_{(t;x,v)}(X^{\delta}_{T}-X^{0}_{T})^{2} \leq C_{0}\delta\,.$$
	\end{proposition}
	
	\begin{proof}  Since $X^{\delta}_{t},$ $X^{0}_{t}$ solve (\ref{deldX}), (\ref{0dX}), we have
		\begin{equation*}
			X^{\delta}_{T}=x +\int^{T}_{t}rX^{\delta}_{s}ds +\int^{T}_{t}qe^{V_{s}}X^{\delta}_{s}dW_{s}^{1},
		\end{equation*}
		and
		\begin{equation*}
			X^{0}_{T}=x +\int^{T}_{t}rX^{0}_{s}ds +\int^{T}_{t}qe^{v}X^{0}_{s}dW_{s}^{1}\,,
		\end{equation*}
		which can be combined to give 
		\begin{eqnarray*}
			X^{\delta}_{T}-X^{0}_{T} &=&\int^{T}_{t}r(X^{\delta}_{s}-X^{0}_{s})ds +\int^{T}_{t}q(e^{V_{s}}X^{\delta}_{s}-e^{v}X^{0}_{s})dW_{s}^{1}\\
			&=&\int^{T}_{t}r(X^{\delta}_{s}-X^{0}_{s})ds +\int^{T}_{t}qe^{v}(X^{\delta}_{s}-X^{0}_{s})dW_{s}^{1} +\int^{T}_{t}q(e^{V_{s}} -e^{v})X^{\delta}_{s}dW_{s}^{1}\,.
		\end{eqnarray*}
		Now let $Y_{s}=X^{\delta}_{s}-X^{0}_{s}$, then $Y_{t} = 0$ and
		\begin{equation*}
			Y_{T} =\int^{T}_{t}rY_{s}ds +\int^{T}_{t}qe^{v}Y_{s}dW_{s}^{1} +\int^{T}_{t}q(e^{V_{s}}-e^{v})X^{\delta}_{s}dW_{s}^{1}.
		\end{equation*}
		Thus,
		\begin{eqnarray*}
			E_{(t;x,v)}[Y^{2}_{T}]&\leq& 3E_{(t,x,v)}\left[\left(\int^{T}_{t}rY_{s}ds\right)^{2} +\left(\int^{T}_{t}qe^{v}Y_{s}dW_{s}^{1}\right)^{2} +\left(\int^{T}_{t}q(e^{V_{s}} -e^{v})X^{\delta}_{s}dW_{s}^{1}\right)^{2}\right]\\
			&\leq & \int^{T}_{t}(3Tr^{2}+3\sigma_{\max}^{2}e^{2v})E_{(t,x,v)}[Y^{2}_{s}]ds + \underbrace{3\sigma_{\max}^{2}\int^{T}_{t}E_{(t;x,v)}\left[(e^{V_{s}} -e^{v})^{2}(X^{\delta}_{s})^{2}\right]ds}_{R(\delta)}\,.
		\end{eqnarray*}
		We have seen before that $X_{t}$ and $V_{t}$ have uniformly bounded moments for $\delta$ sufficiently small.
		We can therefore show that $|R(\delta)|\leq C\delta$ for $C$ independent of $\delta$.
		Setting $q=\sigma_{\max}$ and using Gronwall's inequality, the previous inequality can be recast as 
		$$f(T) \leq\int^{T}_{t}\lambda f(s)ds + C\delta \leq \delta\int^{T}_{t}C\lambda e^{\lambda(T-s)}ds + C\delta \,,$$
		where 
		$f(T)=E_{(t;x,v)}(Y^{2}_{T})$ and $\lambda= 3Tr^{2}+3\sigma_{\max}^{2}e^{2v}>0$. 
		As a consequence,
		$$E_{(t;x,v)}(X^{\delta}_{T} -X^{0}_{T})^{2} = E_{(t;x,v)}Y^{2}_{T} = f(T) \leq C_{0}\delta\,.$$
	\end{proof}
	

	\begin{theorem}\label{theo1}
		The function $P^{\delta}$ uniformly converges to $P^{0}$ with rate $\sqrt{\delta}$ as $\delta\to 0$, where the  convergence is uniform on any compact subset of $[0,T]\times\mathbb{R}\times\mathbb{R}^{+}$, 
	\end{theorem}

	\begin{proof}
		Due to the Lipschitz continuity of $h$, the Cauchy-Schwartz inequality and Proposition\thinspace \ref{prop3}, we get
		\begin{eqnarray*}
			|P^{\delta} -P^{0}|&=&\exp(-r(T-t))\left|\sup_{q\in\Theta} E_{(t;x,v)}[h(X^{\delta}_{T})]-\!\!\sup_{q\in\Theta} E_{(t;x,v)}[h(X^{0}_{T})]\right|,\\
			&\leq& \exp(-r(T-t))\sup_{q\in\Theta}\left|E_{(t;x,v)}[h(X^{\delta}_{T})]-E_{(t;x,v)}[h(X^{0}_{T})]\right|,\\
			&\leq& \exp(-r(T-t))\sup_{q\in\Theta} E_{(t;x,v)}\left|h(X^{\delta}_{T})-h(X^{0}_{T})\right|,\\
			&\leq& K_{0} \exp(-r(T-t))\sup_{q\in\Theta} E_{(t;x,v)}\left|X^{\delta}_{T} -X^{0}_{T}\right|,\\
			&\leq& K_{0} \exp(-r(T-t))\sup_{q\in\Theta}\left[E_{(t;x,v)}(X^{\delta}_{T} -X^{0}_{T})^{2}\right]^{1/2}.
		\end{eqnarray*}
		This entails
		\[
		|P^{\delta} -P^{0}|\leq C_{1}\sqrt{\delta}\,
		\]
		and concludes the proof.
	\end{proof}
	\subsection{Pricing G-PDE}
	The worst-case scenario price $P^{\delta}$ is the solution to the following Hamilton-Jacobi-Bellman $(HJB)$ equation with terminal condition $P^{\delta}(T;x,v) = h(x)$ (see \cite{linos1983optimal,lions1983optimal}): 
	\begin{equation}\label{NPDE}
		-\partial_{t} P^{\delta} = r\left( x\partial_{x}P^{\delta}-P^{\delta} \right)+\sup_{q\in\Theta}\left \{ \frac{1}{2} x^{2}q^{2}e^{2v}\partial_{xx}^{2}P^{\delta}+\sqrt{\delta}qxe^{v}\sigma\rho \partial_{xv}^{2}P^{\delta}\right \} 
		+\delta( \frac{1}{2} \sigma^{2} \partial_{vv}^{2}P^{\delta} + (a-b e^{\alpha v}) \partial_{v}P^{\delta}),
	\end{equation}
	Throughout the rest of the paper, we set $r=0$, i.e. we asusme that the return of the asset is zero, but the return of the option depends on the volatility. In other words, even though the financial asset has no return, the option can have it.\\
	
	\noindent
	\textbf{Leading order term $P_{0}$:} To approximate the value function $P^{\delta}$, we use the regular perturbation expansion
	\begin{equation}\label{erroP}
		P^{\delta}=P_{0}+\sqrt{\delta}P_{1} + \delta P_{2}+\dots,
	\end{equation}
	where $P_{0}$ the leading order term and $P_{1}:=P_{1}(t,x,v)$ the first correction for the approximation of the worst-case scenario price $P^{\delta}$. 
	Substituting (\ref{erroP}) in (\ref{NPDE}), and using Theorem\thinspace\ref{theo1}, the leading order term $P_{0}$ is found to be the solution to
	\begin{equation}\label{EDP0}
		- \partial_{t}P_{0}  =  \sup _{q\in\Theta}\left\{\frac{1}{2}q^{2}e^{2v}x^{2}\partial^{2}_{xx}P_{0}\right\}\,,\quad P_{0}(T;x,v) = h(x),
	\end{equation}
	
	\subsection{Convergence of the second partial derivative}
	The gamma $\partial^{2}_{xx}P^{\delta}$ represents the convexity of the price of an option according to the price of the underlying asset. It indicates whether the price of the option tends to move faster or slower than the price of the underlying asset. Using the fact that $q \in [\sigma_{\min}, \sigma_{\max}]$ , and the regularity results for uniformly parabolic equations which are referenced in \cite{crandall2000lp},\cite{guyon2013nonlinear}, we conclude that  (\ref{NPDE}) is uniformly parabolic.
	
	\begin{proposition}\label{prop4}
		As $\delta\to 0$, the second partial derivative $\partial^{2}_{xx}P^{\delta}$ converges uniformly to $\partial^{2}_{xx}P_{0}$ on any compact subset of $[0,T]\times\mathbb{R}\times\mathbb{R}^{+}$ and with rate $\sqrt\delta$.
	\end{proposition}
	\begin{proof}
		The function $h\in C^4$ is gradient Lipschitz and satisfies polynomial growth conditions in its first four derivatives. By \cite[Thm.~5.2.5]{giga2010nonlinear}, we conclude 
		\begin{itemize}
			\item $
			P^{\delta}(t,.,.)\in C^{1,2,2}_{p}  \mbox{ for } \delta \mbox{ fixed }$
			\item $
			\partial_{x} P^{\delta}(t,.,.) \mbox{ and } \partial^{2}_{xx} P^{\delta}(t,.,.) \mbox{ are uniformly bounded in } \delta
			$
		\end{itemize}
		The assertion thus follows from Theorem\thinspace\ref{theo1}.
	\end{proof}
	
	\noindent
	\textbf{Optimal controls:} Following \cite{fouque2018uncertain}, we define $S^{0}_{t,v}$ to be the zero level set of $\partial^{2}_{xx}P_{0}$
	and the set $A^{\delta}_{t,v}$ to be the set on which $\partial^{2}_{xx}P^{\delta}$ and $\partial^{2}_{xx}P_{0}$ have  different signs, i.e. 
	$$S^{0}_{t,v} := \{x = x(t,v) \in\mathbb{R}^{+}|\partial^{2}_{xx}P_{0}(t;x,v) = 0\}.$$
	and 
	\begin{equation}\label{A}
		A^{\delta}_{t,v} :=\{x = x(t,v)|\partial^{2}_{xx}P^{\delta}(t;x,v) > 0 , \partial^{2}_{xx}P_{0}(t;x,v)< 0\}\,. 
	\end{equation}
	
	\begin{lemma}\label{lem2}
		Call 
		\begin{eqnarray}\label{qdelta}
			q^{*,\delta}(t;x,v) :=\arg{\max}_{q\in\Theta}\left\{\frac{1}{2}q^{2}e^{2v}x^{2}\partial^{2}_{xx}P^{\delta}+\sqrt{\delta}(q\rho\sigma e^{v}x\partial^{2}_{xv}P^{\delta})\right\},
		\end{eqnarray}
		for $x\not\in S^{0 }_{t,v}$ and $\delta>0$ sufficiently small, and 
		\begin{eqnarray} \label{qq0}
			q^{*,0}(t;x,v) :=\arg{\max}_{q\in\Theta}\left\{\frac{1}{2}q^{2}e^{2v}x^{2}\partial^{2}_{xx}P^{0}\right\},
		\end{eqnarray}
		for $\delta=0$. Moreover, let (\ref{qdelta}) and (\ref{qq0}) denote the optimal controls in the G-PDE (\ref{NPDE}) for   $P^{\delta}$ and in the G-PDE (\ref{EDP0}) for $P_{0}$, respectively. Then the limiting optimal control as $\delta\to 0$ is given by 
		\begin{equation}\label{q0}
			q^{*,0}(t;x,v) =
			\begin{cases}
				\sigma_{\max}\,, & \partial^{2}_{xx}P_{0}\geq0,\\
				\sigma_{\min}\,, & \partial^{2}_{xx}P_{0}<0.
			\end{cases}
		\end{equation}
	\end{lemma}
	\begin{proof}
		Let $$f(q):=\frac{1}{2}q^{2}e^{2v}x^{2}\partial^{2}_{xx}P^{\delta}+\sqrt{\delta}(q\rho\sigma e^{v}x\partial^{2}_{xv}P^{\delta}).$$
		and suppose that the maximiser $\hat{q}^{*,\delta}$ is in the interior of the interval $[\sigma_{\min},\sigma_{\max}]$. 
		Then, 
		for $x\not\in S^{0}_{t,v}$, we have 
		$$\hat{q}^{*,\delta}=\frac{-\rho\sqrt\delta\sigma\partial^{2}_{xv}P^{\delta}}{xe^{v}\partial^{2}_{xx}P^{\delta}}.$$
		for the maximiser of $f(q)$. But since $f(\hat{q}^{*,\delta})\to 0$ as $\delta\to 0$, the maximiser must be on the boundary whenever $\delta$ is sufficiently small. In this case, since the sign of $\partial^{2}_{xx}P^{\delta}$ determines the sign of the coefficient of the $q^{2}$ term in $f(q)$, we have $q^{*,\delta}\to q^{*,0}$ pointwise on $S^{0}_{t,v}$ where, for any sufficiently small $\delta\ge 0$, the maximiser can be represented by  
		$$q^{*,\delta}=\sigma_{\max}\mathbf{1}_{\{\partial^{2}_{xx}P^{\delta}\geq 0\}}+\sigma_{\min}\mathbf{1}_{\{\partial^{2}_{xx}P^{\delta}<0\}}.$$
	\end{proof}
	
	Lemma\thinspace\ref{lem2} allows us to rewrite the G-HJB equation (\ref{NPDE}) as
	
	\begin{equation}\label{EDPqdel}
		= \partial_{t}P^{\delta} = \frac{1}{2}(q^{*,\delta})^{2}e^{2v}x^{2}\partial^{2}_{xx}P^{\delta}
		+\sqrt\delta(q^{*,\delta}\rho\sigma e^{v}x\partial^{2}_{xv}P^{\delta})+\delta(\frac{1}{2} \sigma^{2}\partial^{2}_{vv}P^{\delta}+ (a-b e^{\alpha v})\partial_{v}P^{\delta}),
	\end{equation}
	with terminal condition $P^{\delta}(T;x,v) = h(x)$ and with $q^{*,\delta}$ as given above.

	\subsection{First-order corrector for the limit payoff}
	We will now derive a corrector result for the difference $P^\delta-P^0$. To this end, 
	recall that $P_{1}$, the first order correction term of $P^{\delta}$, is the solution to the linear equation
	\begin{equation}\label{EDP1}
		- \partial_{t}P_{1}  = \frac{1}{2}(q^{*,0})^{2}e^{2v}x^{2}\partial^{2}_{xx}P_{1} + q^{*,0}\rho\sigma e^{v}x\partial^{2}_{xv}P_{0}\,,\quad P_{1}(T,x,v) =0\,,
	\end{equation}
	where $q^{*,0}$ is given by (\ref{q0}). Further recall that vanna $\partial^{2}_{xv}P^{\delta}$ is a second order derivative of the option, once to the underlying asset price and once to volatility. It is the sensitivity of the option delta with respect to change in volatility, or, alternatively, the it is the sensitivity of vega $\partial^{2}_{v}P^{\delta}$  with respect to the underlying asset price. For more details see section~4.2.4 in \cite{fouque2011multiscale}

	In the following part we will exploit results from \cite{fouque2003singular} and \cite{fouque2011multiscale} to show that, under the regularity conditions imposed on the derivatives of $h$, the pointwise approximation error $|P^{\delta} -P_{0} -\sqrt{\delta}P_{1}|$ is indeed of order $\mathcal{O}(\delta)$.
	\begin{theorem}
		$\forall(t;x,v) \in [0,T]\times \mathbb{R}^{+} \times \mathbb{R}^{+}$, $\exists  C>0$, such that $$|E^{\delta}(t;x,v)| := |P^{\delta}(t;x,v)-P_{0}(t;x,v)-\sqrt{\delta}P_{1}(t;x,v)|\leq C\delta,$$ where $C$ may depend on $(t;x,v)$ but not on $\delta$.
	\end{theorem}
	
	\begin{proof}
		Adopting the arguments of Secs.~1.9.3 and 4.1.2 in \cite{fouque2011multiscale}, we define the following linear parabolic differential operator
		\begin{equation}\label{Ldelq}
			\begin{aligned}
				\mathcal{L}^{\delta}(q) := & \partial_{t}+\frac{1}{2}
				q^{2}e^{2v}x^{2}\partial^{2}_{xx} +\sqrt{\delta}q\rho e^{v}x\partial^{2}_{xv} + \delta(\frac{1}{2}\sigma^{2}\partial^{2}_{vv} +(a-be^{\alpha v})\partial_{v})\\
				= & \mathcal{L}_{0}(q) + \sqrt{\delta}\mathcal{L}_{1}(q) + \delta\mathcal{L}_{2},
			\end{aligned}
		\end{equation}
		where $\mathcal{L}_{0}(q)$ contains the time derivative and the Black-Scholes operator, $\mathcal{L}_{1}(q)$ contains the mixed derivative due to the covariation between $X_{t}$ and $V_{t}$, and $\delta\mathcal{L}_{2}$ is the infinitesimal generator of the volatility process $V_{t}$.\\
		
		\noindent 
		We can recast equation (\ref{EDPqdel}) as
		\begin{equation}\label{LdelqdelPdel}
			\begin{array}{cc}
				\mathcal{L}^{\delta}(q^{*,\delta})P^{\delta}=0,\\
				P^{\delta}(t;x,v)=h(x).
			\end{array}
		\end{equation}
		Equivalently, equation (\ref{EDP0}) reads  
		\begin{equation}\label{L0q0p0}
			\begin{array}{cc}
				\mathcal{L}_{0}(q^{*,0})P_{0}=0,\\
				P_{0}(T;x,v)=h(x)\,.
			\end{array}
		\end{equation}
		and (\ref{EDP1}) can be expressed by 
		\begin{equation}\label{L0q0P1}
			\begin{array}{cc}
				\mathcal{L}_{0}(q^{*,0})P_{1}+ \mathcal{L}_{1}(q^{*,0})P_{0},\\
				P_{1}(T,x,v)=h(x).
			\end{array}
		\end{equation}
		Now, applying the operator $\mathcal{L}^{\delta}(q^{*,\delta})$ to the error term $E^{\delta} = P^{\delta}-P_{0}-\sqrt{\delta}P_{1}$, we obtain 
		\begin{eqnarray*}
			\mathcal{L}^{\delta}(q^{*,\delta})E^{\delta} & = & \mathcal{L}^{\delta}(q^{*,\delta})(P^{\delta}-P_{0}-\sqrt{\delta}P_{1}) \\
			& = & -(\mathcal{L}_{0}(q^{*,\delta}) + \sqrt{\delta}\mathcal{L}_{1}(q^{*,\delta}) + \delta\mathcal{L}_{2}q^{*,\delta}))(P_{0}+\sqrt{\delta}P_{1})\\
			& = & -\underbrace{\sqrt{\delta}\mathcal{L}_0(q^{*,\delta}) P_1 +  \sqrt{\delta}\mathcal{L}_1(q^{*,\delta}) P_0}_{=0} - \delta\mathcal{L}_2(q^{*,\delta}) P_0 + \delta\mathcal{L}_1(q^{*,\delta}) P_1 +  \delta^{3/2}\mathcal{L}_2(q^{*,\delta}) P_1   
		\end{eqnarray*}
		Using the terminal condition
		$$E^{\delta}(T;x,v) = P^{\delta}(T;x,v)-P_{0}(T;x,v)-\sqrt{\delta}P_{1}(T;x,v) =0\,$$
		and the continuity of the solution to the parabolic equation (\ref{EDP1}), we conclude that $|E^{\delta}(t;x,v) |=\mathcal{O}(\delta)$. 
	\end{proof}
	
	\noindent
	\textbf{Feynman-Kac representation of the error term:} Now recall that the asset price in the worst-case scenario is governed by (\ref{deldX}) with $r=0$ and $q=q^{*,\delta}$: 
	\begin{eqnarray}\label{2}
		dX_{t}^{*,\delta}=q^{*,\delta}_t e^{V_{t}}X_{t}^{*,\delta}dW_{t}^{1},
	\end{eqnarray}
	where, by Lemma\thinspace\ref{lem2}, the optimal control $(q_{t})=(q^{*,\delta})$ is explicitly given  for sufficiently small $\delta$. (It is straighforward to establish the existence and the uniqueness of the solution of (\ref{2}) $X_{t}^{*,\delta}$.)\\
	
	\noindent 
	We can apply the Feynman-Kac formula to get probabilistic representation of $E^{\delta}(t,x,v)$, namely,
	$$E^{\delta}(t,x,v) = I_{0} + \delta^{\frac{1}{2}} I_{1} + \delta I_{2} + \delta^{\frac{3}{2}} I_{3},$$
	where
	$$
	I_{0}= \mathbb{E}_{(t,x,v)} \left[\int^{T}_{t}
	\frac{1}{2}\left((q^{*,\delta})^{2} -(q^{*,0})^{2}\right)e^{2V_{s}}(X^{*,\delta}_{s})^{2}\partial^{2}_{xx}P_{0}(s,X^{*,\delta}_{s} ,V_{s})ds\right],\qquad
	$$
	
	$$
	I_{1} = \mathbb{E}_{(t,x,v)}\left[\int^{T}_{t} (q^{*,\delta}-q^{*,0})\rho\sigma e^{V_{s}}X^{*,\delta}_{s} \partial^{2}_{xv}P_{0}(s,X^{*,\delta}_{s} ,V_{s})\right.\qquad\qquad\qquad\qquad$$
	$$
	\left.+ \frac{1}{2}\left((q^{*,\delta})^{2}-(q^{*,0})^2\right)e^{2V_{s}}(X^{*,\delta}_{s})^{2}  \partial^{2}_{xx}P_{1}(s,X^{*,\delta}_{s},V_{s})ds\right],
	$$
	
	$$
	I_{2}= \mathbb{E}_{(t,x,v)}\left[\int^{T}_{t}q^{*,\delta}\rho\sigma e^{V_{s}}X^{*,\delta}_{s}\partial^{2}_{xv}P_{1}(s,X^{*,\delta}_{s} ,V_{s}) +
	\frac{1}{2} \sigma^{2}\partial^{2}_{vv}P_{0}(s,X^{*,\delta}_{s} ,V_{s}) \right.$$
	$$
	\left.
	+(a-be^{\alpha V_{s}})\partial_{v}P_{0}(s,X^{*,\delta}_{s},V_{s})ds\right],
	$$
	
	$$
	I_{3}= \mathbb{E}_{(t,x,v)}\left[\int^{T}_{t}\frac{1}{2}
	\sigma^{2}\partial^{2}_{vv}P_{1}(s,X^{*,\delta}_{s},V_{s}) +(a-be^{\alpha V_{s}})\partial_{v}P_{1}(s,X^{*,\delta}_{s} ,V_{s})ds\right].
	$$
	Noting that 
	\begin{eqnarray*}
		\{q^{*,\delta}{\not =} q^{*,0}\} &=& \mathcal{A}^{\delta}_{t,v},  \\
		q^{*,\delta}-q^{*,0}&=& (\sigma_{\max}-\sigma_{\min})(\mathbf{1}_{\{\partial^{2}_{xx}P^{\delta}\geq 0\}}-\mathbf{1}_{\{\partial^{2}_{xx}P_{0}\geq 0\}}),\\
		\mbox{and}\qquad\qquad
		(q^{*,\delta})^{2}-(q^{*,0})^{2}&=&(\sigma_{\max}^{2}-\sigma_{\min}^{2})(\mathbf{1}_{\{\partial^{2}_{xx}P^{\delta}\geq 0\}}-\mathbf{1}_{\{\partial^{2}_{xx}P_{0}\geq 0\}})\,
	\end{eqnarray*}
	the next theorem shows that $I_{0}$ , $I_{1}$ are indeed of order $\mathcal{O}(\delta)$ and  $\mathcal{O}(\sqrt{\delta})$.
	\begin{theorem}
		There exist constants $M_{0}, M_{1}>0$ depending on $(t,x,v)$, but not on $\delta$, such that 		$$|I_{0}|\leq M_{0}\delta\,,\quad\text{and}\quad |I_{1}|\leq M_{1}\sqrt{\delta}\,.$$
	\end{theorem}
	
	\begin{proof}
		The proof follows the same method as in \cite{fouque2018uncertain}.
	\end{proof}

	\section{Second-order BSDE representation of the worst-case scenario}\label{sec:2BSDE}
	
	We recall the definition of 2BSDE, and we will explain how it is linked to our G-HJB equation; for details, we refer to  \cite{cheridito2007second}.
	\begin{definition}
		Let $(t,x) \in [0,T) \times \mathbb{R}^d$, $(X_s^{t,x})_{s \in [t,T]}$ a diffusion process and $(Y_s,Z_s, \Gamma_s, A_s)_{s \in [t,T]}$ a quadruple of $\mathbb{F}^{t,T}$-progressively measurable processes
		taking values in $\mathbb{R}$, $\mathbb{R}^d$, $\mathcal{S}^d$ and $\mathbb{R}^d$, respectively. The quadruple $(Y,Z,\Gamma,A)$ is called a solution to the second order backward stochastic differential equation {\rm (2BSDE)} corresponding to $(X^{t,x},f,g)$ if
		\begin{eqnarray}
			\label{2bsde1}
			dY_s &=& f(s, X^{t,x}_s, Y_s ,Z_s , \Gamma_s) \,ds
			+ Z_s' \circ dX^{t,x}_s \, , \quad s \in [t,T) \, ,\\
			\label{2bsde2}
			dZ_s &=& A_s \,ds + \Gamma_s \,dX^{t,x}_s \, , \quad s \in [t,T) \, ,\\
			\label{2bsde3}
			Y_T &=& g\left(X^{t,x}_T\right) \, ,
		\end{eqnarray}
		where $Z_s' \circ dX^{t,x}_s$ denotes Fisk--Stratonovich
		integration, which is related to It\^{o} integration by
		$$
		Z_s' \circ dX^{t,x}_s =
		Z_s' \,dX^{t,x}_s + \frac{1}{2} \,d\left<{Z, X^{t,x}}_s\right>
		=
		Z_s' \,dX^{t,x}_s + \frac{1}{2} \,{\rm Tr}
		[\Gamma_s \sigma(X^{t,x}_s) \sigma(X^{t,x}_s)' ] \,ds \, .
		$$
	\end{definition}
	The last definition furnishes a fundamental relation between 2BSDE like (\ref{2bsde1})-(\ref{2bsde3}) and fully nonlinear parabolic PDEs. To understand this relation, let $f : [0,T) \times \mathbb{R}^d \times
	\mathbb{R} \times \mathbb{R}^d \times \mathcal{S}^d \to \mathbb{R}$
	and $g : \mathbb{R}^d \to \mathbb{R}$ be continuous functions. Further assume that $u: [0,T] \times \mathbb{R}^d \to \mathbb{R}$ is a continuous function with the properties  
	$$u_t, Du, D^2u, \mathcal{L} D u \in \mathcal{C}^0([0,T) \times \mathbb{R}^d)\,,$$
	that solves the PDE
	\begin{eqnarray}\label{pde}
		- u_t(t,x) + f\left(t,x,u(t,x),Du(t,x),D^2u(t,x)\right)
		&=& 0 \quad \mbox{on } [0,T) \times \mathbb{R}^d \, ,
	\end{eqnarray}
	with terminal condition
	\begin{equation} \label{terminal}
		u(T,x) =g(x) \, , \quad x \in \mathbb{R}^d \, .
	\end{equation}
	Then, it follows directly from It\^{o}'s formula that
	for each pair $(t,x) \in [0,T) \times \mathbb{R}^d$, the
	processes
	\begin{eqnarray*}
		Y_s &=& u\left(s,X^{t,x}_s\right) \, , \quad s \in [t,T] \, ,\\
		Z_s &=& Du\left(s, X^{t,x}_s\right) \, , \quad s \in [t,T] \, ,\\
		\Gamma_s &=& D^2u\left(s, X^{t,x}_s\right) \, , \quad s \in [t,T] \, ,\\
		A_s &=& \mathcal{L} Du \left(s, X^{t,x}_s\right) \, , \quad s \in [t,T] \, ,
	\end{eqnarray*}
	solve the 2BSDE corresponding to $(X^{t,x}, f,g)$. Conversely, the first component of the solution of the 2BSDE \eqref{2bsde1} at the initial time is a solution of the fully nonlinear PDE \eqref{pde} satisfies $Y_t = u(t,x)$. Note that the representation of  \eqref{pde} by a 2BSDE is not unique, even though its solution is (cf.~\cite{hartmann2019chaos}). 
	
	The  representation of fully nonlinear parabolic PDEs, such as (\ref{EDPqdel}), allows to solve them numerically by solving the corresponding 2BSDE, e.g. by using the techniques described in \cite{beck2019machine}.

	\subsection{2BSDE representation of the payoff}
	Here we specifically use the link between our G-HJB equation and 2BSDEs to improve the convergence rate of the convergence $P^\delta\to P^0$. To this end we write the 2BSDE for $P_\delta$ (resp. $P_0$) as follows: for all $s \in [t,T)$ it holds that 
	\begin{eqnarray}
		\label{2bsde1EXAMPLE}
		dY^{\delta;t,x}_s &=& f^{\delta}(s, \tilde{X}^{\delta;t,x}_s, Y^{\delta;t,x}_s ,Z^{\delta;t,x}_s , \Gamma^{\delta;t,x}_s) \,ds
		+ (Z^{\delta;t,x})_s' \circ d\tilde{X}^{\delta;t,x}_s,\\
		\label{2bsde2EXAMPLE}
		dZ^{\delta;t,x}_s &=& A^{\delta}_s \,ds + \Gamma^{\delta}_s \,d\tilde{X}^{\delta;t,x}_s,\\
		\label{2bsde3EXAMPLE}
		Y^{\delta;t,x}_T &=& h\left(\tilde{X}^{\delta;t,x}_T\right) ,
	\end{eqnarray}
	where $\tilde{X}$ is the solution to the SDE 
	$$d(X_{t}^{\delta},V_{t})=d\tilde{X}_t=d\tilde{W}_t,\quad d\tilde{W}_t=d(W^1_{t},W^2_{t}), \quad \tilde{X}_0=\tilde{x}$$
	Similarly, 
	\begin{eqnarray}
		\label{2bsde0EXAMPLE}
		dY^{0;t,x}_s &=& f^{0}(s, {X}^{0;t,x}_s, Y^{0;t,x}_s ,Z^{0;t,x}_s , \Gamma^{0}_s) \,ds
		+ (Z^{0;t,x})_s' \circ d{X}^{0;t,x}_s,\\
		\label{2bsde02EXAMPLE}
		dZ^{0;t,x}_s &=& A^{0}_s \,ds + \Gamma^{0}_s \,d{X}^{0;t,x}_s ,\\
		\label{2bsde03EXAMPLE}
		Y^{0;t,x}_T &=& h\left({X}^{0;t,x}_T\right),
	\end{eqnarray}
	where $X_{t}^{0}$ is the solution to 
	$$dX^{0}_{t} = dW^{1}_{t}, \quad X_{0} = x$$. Here $h$ denotes the payoff function (specified below), and 
	$$f^{0}(s,{x},y,z,S)=-\frac{1}{2}{x}^{0}e^{2v}|\bar{\sigma}(S_{1,1})|^2 S_{1,1}$$
	$$f^{\delta}(s,\tilde{x},y,z,S)=-\frac{1}{2}\tilde{x}^{\delta}e^{2v}|\bar{\sigma}(S_{1,1})|^2 S_{1,1}-2\sqrt{\delta}\tilde{x}^{\delta}e^{v}\sigma \rho |\bar{\sigma}(S_{1,2})|S_{1,2}
	-\delta \left( \frac{1}{2}\sigma^2 S_{2,2}+(a-be^{\alpha v})z_2 \right),$$
	where,
	\begin{equation}
		\bar{\sigma}=	
		\begin{cases} 
			\sigma_{\max}\, & x\geq 0\\
			\sigma_{\min}\, & x< 0
		\end{cases}.
	\end{equation}
	Note that the nonlinear diffusion coefficient has been moved to the drift terms (or: drivers) $f^0$ and $f^\delta$, which is why the SDE dynamics is trivial.  
	Then from the link between G-PDEs and 2BSDEs we have $Y^{0;t,x}_t=P_0(t,x)$ and $Y^{\delta;t,x}_t=P_\delta(t,x)$. \\
	
	\noindent 
	We will now  use this link to revisit the convergence result for $P_\delta \to P_0$.
	\begin{theorem} \label{conv}
		$P_{\delta}$ converges to $P_{0}$ as $\delta\to 0$, uniformly on compact sets and at rate ${\delta}$. 
	\end{theorem}
	\begin{proof}
		We have
		\begin{eqnarray*}
			Y^{\delta;t,x}_{t} &=& h(\tilde{X}_{T}^{\delta;t,x})+\int_{t}^{T}f^{\delta}(r, \tilde{X}^{\delta;s,x}_{r},Y^{\delta;s,x}_{r}, Z^{\delta;s,x}_{r},\Gamma^{\delta;s,x}_{r})dr -\int_{t}^{T} (Z^{\delta;s,x})'_{r} \circ d\tilde{X}_{r}^{\delta;s,x},\\
			(Z^{\delta;s,x})_{r}'\circ d\tilde{X}^{\delta;s,x}_{r}&=&(Z^{\delta;s,x})_{r}'d\tilde{X}^{\delta;s,x}_{r}+\frac{1}{2}{\rm Tr}[\Gamma^{\delta}_{r}\sigma(\tilde{X}^{\delta;s,x}_{r})\sigma(\tilde{X}^{\delta;s,x}_{r})']dr,\\
			(Z^{\delta;s,x})_{r}'d\tilde{X}^{\delta;s,x}_{r}&=&(Z_{1}^{\delta;t,x})_{r}'dW^{1}_{r}+(Z_{2}^{\delta;s,x})_{r}'dW^{2}_{r},
		\end{eqnarray*}
		and thus
		$$Y^{\delta;t,x}_{t} = h(\tilde{X}_{T}^{\delta;t,x})+\int_{t}^{T}f^{\delta}(r, \tilde{X}^{\delta;s,x}_{r},Y^{\delta;s,x}_{r}, Z^{\delta;s,x}_{r},\Gamma^{\delta;s,x}_{r})dr-\int_{t} ^{T}((Z_{1}^{\delta;s,x})_{r}'dW^{1}_{r}+(Z_{2}^{\delta;s,x})_{r}'dW^{2}_{r})$$
		$$-\int_{t}^{T}\frac{1}{2}{\rm Tr}[\Gamma^{\delta}_{r}\sigma(\tilde{X}^{\delta;s,x}_{r})\sigma(\tilde{X}^{\delta;s,x}_{r})']dr.$$
		\begin{eqnarray*}
			Y^{0;t,x}_{t}& =& h(X^{0;t,x})+\int_{t}^{T} f^{0}(r, X^{0;s,x}_{r}, Y^{0;s,x}_{r},Z^{0;s,x}_{r}, \Gamma^{0;s,x}_{r})dr-\int_{t} ^{T} (Z^{0;s,x})_{r}' \circ dX^{0;s,x}_{r},\\
			(Z^{0;s,x})_{r}' \circ d\tilde{X}^{0;s,x}_{r} &=& (Z^{0;s,x})_{r}'d\tilde{X}^{0;s,x}_{r}+\frac{1}{2}{\rm Tr}[\Gamma^{0}_{r}\sigma(\tilde{X}^{0;s,x}_{r})\sigma(\tilde{X}^{0;s,x}_{r})']dr\,.
		\end{eqnarray*}
		Calling $\tilde{Z}_{r}^{s,x}=(Z_{r}^{0;s,x},0)$
		$$(Z^{0;s,x})_{r}'d\tilde{X}^{0;s,x}_{r} = (\tilde{Z}_{s})'d\tilde{X}^{\delta;s,x}_{r} = (Z^{0;s,x})_{r}'dW^{1}_{r} + 0\,,$$
		we obtain 
		$$Y^{0;t,x}_{t} = h(X_{T}^{0;t,x})+\int_{t}^{T}{f}^{0}(r, X^{0;s,x}_{r}, Y^{0;s,x}_{r}, Z^{0;s,x}_{r}, \Gamma^{0;s,x}_{r})dr-\int_{t}^{T}(Z^{0;s,x})_{r}'dW^{1}_{r}$$
		$$-\int_{t}^{T}\frac{1}{2}{\rm Tr}[\Gamma^{0}_{r}\sigma(\tilde{X}^{0;s,x}_{r})\sigma(\tilde{X}^{0;s,x}_{r})']dr.$$
		Now let $y_{t} = Y^{\delta;t,x}_{t}- Y^{0;t,x}_{t}$. Then 
		\begin{align*}
			y_{t} = & h(\tilde{X}^{\delta;t,x}_{T})-h(X^{0;t,x}_{T})
			+\int_{t}^{T}{f}^{\delta}(r,\tilde{X}^{\delta;s,x}_{r},Y ^{\delta;s,x}_{r}, Z^{\delta;s,x}_{r}, \Gamma^{\delta;s,x}_{r})-{f}^{0}(r, X^{0;s,x}_{r}, Y^{0;s,x}_{r}, Z^{0;s,x}_{r}, \Gamma^{0;s,x}_{r})dr\\
			&-\int_{t}^{T} (((Z_{1}^{\delta;s,x})_{r}'dW^{1}_{r}+(Z_{2}^{\delta;s,x})_{r}'dW^{2}_{r})-(Z^{0;s,x})_{r}'dW^{1}_{r}) - {\rm Tr}[\Gamma^{0}_{r}\sigma(\tilde{X}^{0;s,x}_{r})\sigma(\tilde{X}^{0;s,x}_{r})']))\\ 
			= & h(\tilde{X}^{\delta;t,x}_{T})-h(X^{0;t,x}_{T}) + \int_{t}^{T}{f}^{\delta}(r, \tilde{X}^{\delta;s,x}_{r}, Y^ {\delta;s,x}_{r}, Z^{\delta;s,x}_{r}, \Gamma^{\delta;s,x}_{r})-{f}^{0}(r, X^{0;s,x}_{r}, Y^{0;s,x}_{r}, Z^{0;s,x}_{r}, \Gamma^{0;s,x}_{r})dr\\
			&-\int_{t}^{T} (((Z_{1}^{\delta;s,x})_{r}'-(Z^{0;s,x})_{r}')dW^{1}_{r} + (Z_{2}^{\delta;s,x})_{r}'dW^{2}_{r})\\ 
			&- \int_{t}^{T}\frac{1}{2}({\rm Tr}[\Gamma^{\delta}_{r}\sigma(\tilde{X}^{\delta;s,x}_{r})\sigma(\tilde{X}^{\delta;s,x}_{r})']-{\rm Tr}[\Gamma^{0}_{r}\sigma(\tilde{X}^{0;s,x}_{r})\sigma(\tilde{X}^{0;s,x}_{r})'])dr,
		\end{align*}
		where 
		\begin{align*} 
			& {f}^{\delta}(r, \tilde{X}^{\delta;s,x}_{r}, Y^ {\delta;s,x}_{r}, Z^{\delta;s,x}_{r}, \Gamma^{\delta;s,x}_{r})-{f}^{0}(r, X^{0;s,x}_{r}, Y^{0;s,x}_{r}, Z^{0;s,x}_{r}, \Gamma^{0;s,x}_{r}) \\ 
			& =-\frac{1}{2}(\tilde{x}^{\delta}e^{2V_{t}}-x^{0}e^{2v})|\sigma(\Gamma_{11})|^{2}\Gamma_{11} -2\sqrt{\delta}\sigma\rho\tilde{x}^{\delta}e^{V_{t}}|\sigma(\Gamma_{12})|\Gamma_{12}-\delta(\frac{1}{2}\sigma^{2}\Gamma_{22}+(a+be^{\alpha V_{t}})z^{\delta}_{2})\,.
		\end{align*}
		Applying It\^{o}'s formula to $e^{\alpha t}|y_{t}|^{2} $ for some $\alpha>0$ then yields 
		\begin{align*}
			d(e^{\alpha t}|y_{t}|^{2}) = & \alpha e^{\alpha s}|y_{s}|^{2}ds\\
			&-2e^{\alpha s}|y_{s}|\{{f}^{\delta}(s, \tilde{X}^{\delta;t,x}_{s}, Y^ {\delta;t,x}_{s}, Z^{\delta;t,x}_{s}, \Gamma^{\delta;t,x}_{s})-{f}^{0}(s, X^{0;t,x}_{s}, Y^{0;t,x}_{s}, Z^{0;t,x}_{s}, \Gamma^{0;t,x}_{s})\}ds\\
			&+2e^{\alpha s}|y_{s}|\{((Z_{1}^{\delta;t,x})_{s}'-(Z^{0;t,x})_{s}')dW^{1}_{s} + (Z_{2}^{\delta;t,x})_{s}'dW^{2}_{s} \}\\\
			&+e^{\alpha s} ({\rm Tr}[\Gamma^{\delta}_{s}\sigma(\tilde{X}^{\delta;t,x}_{s})\sigma(\tilde{X}^{\delta;t,x}_{s})']-{\rm Tr}[\Gamma^{0}_{s}\sigma(\tilde{X}^{0;t,x}_{s})\sigma(\tilde{X}^{0;t,x}_{s})'])ds\\
			&+ e^{\alpha s}\{|| ((Z_{1}^{\delta;t,x})_{s}'-(Z^{0;t,x})_{t}') ||^{2} +|| (Z_{2}^{\delta;t,x})_{s}'||^{2}\}ds.
		\end{align*}
		Therefore,
		\begin{align*}
			e^{\alpha t}|y_{t}|^{2}&+\int_{t}^{T}e^{\alpha r}\{|| ((Z_{1}^{\delta;s,x})_{r}'-(Z^{0;s,x})_{r}') ||^{2} -|| (Z_{2}^{\delta;s,x})_{r}'||^{2}\}dr\\
			&+\int_{t}^{T}e^{\alpha r} ({\rm Tr}[\Gamma^{\delta}_{r}\sigma(\tilde{X}^{\delta;s,x}_{r})\sigma(\tilde{X}^{\delta;s,x}_{r})']-{\rm Tr}[\Gamma^{0}_{r}\sigma(\tilde{X}^{0;s,x}_{r})\sigma(\tilde{X}^{0;s,x}_{r})'])dr,\\
			=& h(\tilde{X}^{\delta;t,x}_{T})-h(X^{0;t,x}_{T})+\int_{t}^{T} e^{\alpha r}(-\alpha)  |y_{s}|^{2}dr\\
			&+\int_{t}^{T}2|y_{s}|\{{f}^{\delta}(r, \tilde{X}^{\delta;s,x}_{r}, Y^ {\delta;s,x}_{r}, Z^{\delta;s,x}_{r}, \Gamma^{\delta;s,x}_{r})-{f}^{0}(r, X^{0;s,x}_{r}, Y^{0;s,x}_{r}, Z^{0;s,x}_{r}, \Gamma^{0;s,x}_{r})\})dr\\
			&-\int_{t}^{T}2e^{\alpha r }|y_{s}|\{((Z_{1}^{\delta;s,x})_{r}'-(Z^{0;s,x})_{r}')dW^{1}_{r} + (Z_{2}^{\delta;s,x})_{r}'dW^{2}_{r} \}.
		\end{align*}
		Since for all $\varepsilon> 0$, we have $2ab \leq a^{2}/\varepsilon + \varepsilon b^{2}$, it follows that 
		\begin{align*}
			e^{\alpha t}|y_{t}|^{2} &+ \int_{t}^{T}e^{\alpha r}\{|| ((Z_{1}^{\delta;s,x})_{r}'-(Z^{0;s,x})_{r}') ||^{2} -|| (Z_{2}^{\delta;s,x})_{r}'||^{2}\}dr\\
			&+\int_{t}^{T}e^{\alpha r} ({\rm Tr}[\Gamma^{\delta}_{r}\sigma(\tilde{X}^{\delta;s,x}_{r})\sigma(\tilde{X}^{\delta;s,x}_{r})']-{\rm Tr}[\Gamma^{0}_{r}\sigma(\tilde{X}^{0;s,x}_{r})\sigma(\tilde{X}^{0;s,x}_{r})'])dr,\\
			\leq & h(\tilde{X}^{\delta;t,x}_{T})-h(X^{0;t,x}_{T})+\int_{t}^{T} e^{\alpha r}(-\alpha  |y_{s}|^{2}dr\\
			&+\int_{t}^{T}(|y{s}|^{2}/\varepsilon+\varepsilon \{{f}^{\delta}(r, \tilde{X}^{\delta;s,x}_{r}, Y^ {\delta;s,x}_{r}, Z^{\delta;s,x}_{r}, \Gamma^{\delta;s,x}_{r})-{f}^{0}(r, X^{0;s,x}_{r}, Y^{0;s,x}_{r}, Z^{0;s,x}_{r}, \Gamma^{0;s,x}_{r})\}^{2})dr\\
			&-\int_{t}^{T}2e^{\alpha r }|y_{s}|\{((Z_{1}^{\delta;s,x})_{r}'-(Z^{0;s,x})_{r}')dW^{1}_{r} + (Z_{2}^{\delta;s,x})_{r}'dW^{2}_{r} \},\,.
		\end{align*}
		Therefore, setting $\alpha=\frac{1}{\varepsilon}$, we conclude 
		\begin{equation}\label{eq37}
			\begin{aligned}
				e^{\alpha t}|y_{t}|^{2}&+\int_{t}^{T}e^{\alpha r}\{|| ((Z_{1}^{\delta;s,x})_{r}'-(Z^{0;s,x})_{r}') ||^{2} +|| (Z_{2}^{\delta;s,x})_{r}'||^{2}\}dr\\
				&+\int_{t}^{T}e^{\alpha r} ({\rm Tr}[\Gamma^{\delta}_{r}\sigma(\tilde{X}^{\delta;s,x}_{r})\sigma(\tilde{X}^{\delta;s,x}_{r})']-{\rm Tr}[\Gamma^{0}_{r}\sigma(\tilde{X}^{0;s,x}_{r})\sigma(\tilde{X}^{0;s,x}_{r})'])dr\\
				\leq & h(\tilde{X}^{\delta;t,x}_{T})-h(X^{0;t,x}_{T})\\
				&+\varepsilon \int_{t}^{T}\{{f}^{\delta}(r, \tilde{X}^{\delta;s,x}_{r}, Y^ {\delta;s,x}_{r}, Z^{\delta;s,x}_{r}, \Gamma^{\delta;s,x}_{r})-{f}^{0}(r, X^{0;s,x}_{r}, Y^{0;s,x}_{r}, Z^{0;s,x}_{r}, \Gamma^{0;s,x}_{r})\}^{2})dr\\
				& -\int_{t}^{T}2e^{\alpha r }|y_{s}|\{((Z_{1}^{\delta;s,x})_{r}'-(Z^{0;s,x})_{r}')dW^{1}_{r} + (Z_{2}^{\delta;s,x})_{r}'dW^{2}_{r} \}.
			\end{aligned}
		\end{equation}
		Because $X_{t}$ and $V_{t}$ have finite moments of any order, the imposed regularity condition on $h$, together with \cite[Thm.~5.2.2]{giga2010nonlinear}, Theorem\thinspace\ref{theo1}, Proposition\thinspace\ref{prop3}, and Proposition\thinspace\ref{prop4} in this paper, imply 
		$$\mathbb{E}(h(\tilde{X}^{\delta;t,x}_{T})-h(X^{0;t,x}_{T}))\leq C\delta,$$
		and 
		$$\mathbb{E}(\{{f}^{\delta}(r, \tilde{X}^{\delta;s,x}_{r}, Y^ {\delta;s,x}_{r}, Z^{\delta;s,x}_{r}, \Gamma^{\delta;s,x}_{r})-{f}^{0}(r, X^{0;s,x}_{r}, Y^{0;s,x}_{r}, Z^{0;s,x}_{r}, \Gamma^{0;s,x}_{r})\}^{2})\leq C_{0}\delta\,.$$
		Hence 
		\begin{align*}
			& \mathbb{E}\left[ \sup_{t\leq s \leq T} e^{\alpha t}|y_{t}|^{2} \right]\\
			& \leq C\delta +C_{0}\varepsilon\delta+C_{1}\mathbb{E}\left[\left( \int_{t}^{T}e^{2\alpha r} |y_{s}|^{2}\{ || ((Z_{1}^{\delta;s,x})_{r}'-(Z^{0;s,x})_{r}') ||^{2} +|| (Z_{2}^{\delta;s,x})_{r}'||^{2}\}dr \right)^{\frac{1}{2}}\right]\\
			& \leq C\delta +C_{0}\varepsilon\delta+C_{1}\mathbb{E}\left[ \sup_{t\leq s \leq T} e^{\alpha t/2}|y_{t}|\left(\int_{t}^{T} e^{\alpha r} \{ || ((Z_{1}^{\delta;s,x})_{r}'-(Z^{0;s,x})_{r}') ||^{2} +|| (Z_{2}^{\delta;s,x})_{r}'||^{2}\}dr  \right)^{\frac{1}{2}} \right], 
		\end{align*}
		which together with the inequality $ab \leq a^{2}/2 + b^{2}/2$ yields 
		\begin{align*}
			& \mathbb{E}\left[ \sup_{t\leq s \leq T} e^{\alpha t}|y_{t}|^{2} \right]\\ 
			& \leq C\delta +C_{0}\varepsilon\delta+ \frac{1}{2}\mathbb{E}\left[\sup_{t\leq s \leq T} e^{\alpha t}|y_{t}|^{2} \right] + \frac{C_{1}^{2}}{2}\mathbb{E}\left[\int_{t}^{T} e^{\alpha r}\{ || ((Z_{1}^{\delta;s,x})_{r}'-(Z^{0;s,x})_{r}') ||^{2} +|| (Z_{2}^{\delta;s,x})_{r}'||^{2}\}dr \right].
		\end{align*}
		As a consequence of the inequality \eqref{eq37}, we thus obtain 
		\begin{align*} 
			& \mathbb{E}\left[ \sup_{t\leq s \leq T} e^{\alpha t}|y_{t}|^{2} +\int_{t}^{T} e^{\alpha r}\{ || ((Z_{1}^{\delta;s,x})_{r}'-(Z^{0;s,x})_{r}') ||^{2} +|| (Z_{2}^{\delta;s,x})_{r}'||^{2}\}dr\right.\\ & \left.+2\int_{t}^{T}e^{\alpha r} ({\rm Tr}[\Gamma^{\delta}_{r}\sigma(\tilde{X}^{\delta;s,x}_{r})\sigma(\tilde{X}^{\delta;s,x}_{r})']-{\rm Tr}[\Gamma^{0}_{r}\sigma(\tilde{X}^{0;s,x}_{r})\sigma(\tilde{X}^{0;s,x}_{r})'])dr \right]\\
			\leq & C\delta +C_{0}\varepsilon\delta+C_{1}^{2}\,,
		\end{align*}
		which entails the final result:  
		$$ \mathbb{E}\left[ \sup_{t\leq s \leq T} e^{\alpha t}|y_{t}|^{2}\right]\leq \delta \tilde{C}_{\varepsilon}$$
		for some $\tilde{C}_{\varepsilon}>0$ independent of $\delta$. 
	\end{proof}
	
	
	
	
	\section{Numerical illustration}\label{sec:numerics}
	
	We conclude with a numerical demonstration of the theoretical results to confirm that $|P_\delta -  P_0|=\mathcal{O}(\delta)$. To this end, note that the valuation of financial derivatives based on our UV model requires solving the G-HJB equation (\ref{NPDE}), which is typically not analytically solvable. 
	
	In low dimension, we can implement a finite difference scheme; here we follow a different route and take advantage of the link between G-PDE and 2BSDE. To be specific the payoff function is chosen as 
	\begin{equation*}\label{TerExample}
		h(x)=(x-90)^{+}-2(x-100)^{+}+(x-110)^{+}
	\end{equation*}
	We consider the following parameters: 
	\[
	\tilde{x}=(\tilde{x}_0,k_0)=(100,-1)\,,\;  \sigma_{\min}=0.1, \sigma_{\max}=0.2\,,\;  \alpha =2\,,\;   T=0.15\,,\;  a=0.6\,, \; b=0.5\,,\; \rho=0.5\,.\]
	For these parameters, we compute the difference between $P_\delta$ and $P_0$, the solutions of the G-PDE \eqref{NPDE} and  \eqref{EDP0}, using the deep learning 2BSDE solver introduced by Beck et al. \cite{beck2019machine}. More specifically, we numerically solve the 2BSDEs \eqref{2bsde1EXAMPLE}-\eqref{2bsde3EXAMPLE} and \eqref{2bsde0EXAMPLE}-\eqref{2bsde03EXAMPLE} with the  Python code provided in \cite{beck2019machine}.  
	
	The result is shown in Table \ref{tab:tab1} and Figure \ref{fig:fig1}. Neglecting the error invoked by the numerical approximation of the deep neural network, which is difficult to assess, the numerical calculation confirms that $|P_\delta -  P_0|\simeq \mathcal{O}(\delta^{0.7})$, which is in agreement with the predictions of Theorem \ref{theo1} and Theorem \ref{conv}.   
	
	\begin{table}
		\begin{tabularx}{0.8\textwidth} {
				| >{\raggedright\arraybackslash}X
				| >{\centering\arraybackslash}X
				| >{\raggedleft\arraybackslash}X
				| >{\raggedleft\arraybackslash}X | }
			\hline
			$\delta$ & 0.5 & 0.2  & 0.001 \\
			\hline
			$error(\delta)$  & 1.2  & 0.6  & 0.02  \\
			\hline
		\end{tabularx}
		\caption{The error $\varepsilon^{0,x}(\delta)=P_\delta(0,x)-P_0(0,x)$ for $\tilde{x}=(100,-1)$.\label{tab:tab1}}
	\end{table}
	
	\begin{figure}[h]
		\begin{center}
			\includegraphics[width=0.75\textwidth]{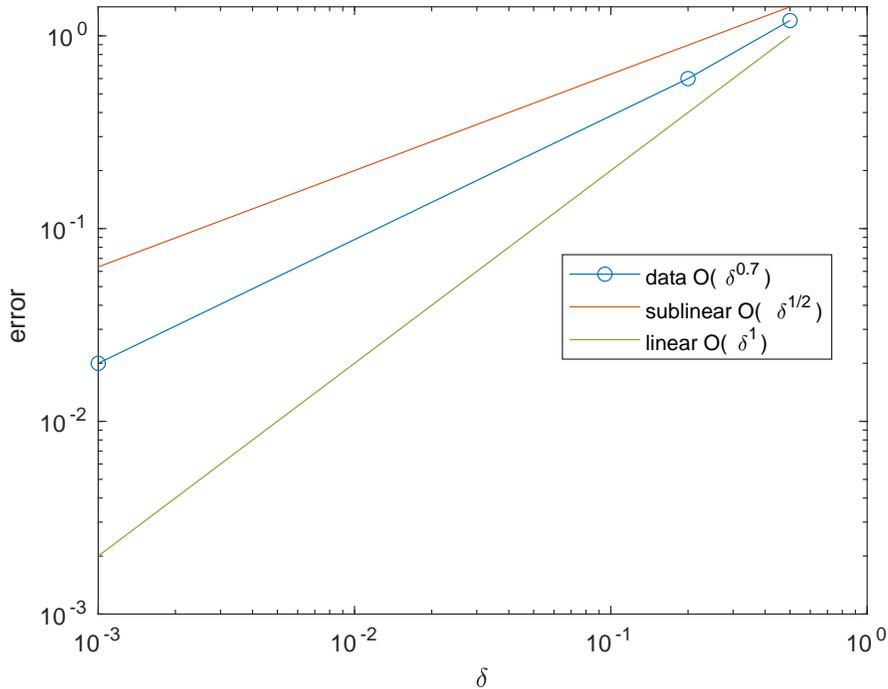}
			\caption{The error $\varepsilon^{0,x}(\delta)=P_\delta(0,x)-P_0(0,x)$ in doubly logarithmic scale; the slope is roughly 0.7}\label{fig:fig1}
		\end{center}
	\end{figure}

	\newpage 
	\section{Conclusion and outlook}\label{sec:conclusion}
	In this work we have studied $\alpha$-hypergeometric stochastic models with uncertain volatility (UV). The idea is to connect the UV model with a nonlinear expectation framework to derive a worst-case price scenario, avoiding the complicated and numerically expensive model calibration step. We have studied the asymptotic behaviour of the worst-case scenario option prices in the case when the time scale at which the stochastic volatility process varies tends to infinity (i.e. when the volatility process becomes infinitely slow). 
	As we have shown, the limit model is an accurate simplified description of the UV model in the regime of the slow variable of the uncertain volatility bounds. The method presented here can be applied also for other models such as the Heston model. 	
	
	We have illustrated our results by a numerical example. The numerical solution of our problem is based on the known link of fully nonlinear second order partial differential equations that describe the worst-case price scenario and second-order backward stochastic differential equations (2BSDEs). We should emphasize that the numerical algorithm we use for solving 2BSDEs even works when the terminal cost that determines the payoff is non-differentiable. Although this paper is only giving a proof of concept, we expect that the ideas can be applied also in the case of UV models when, for example, there is only partial information from the market.
	
	\section{Acknowledgement}
	
	This work has been partially supported by the MATH+ Cluster of Excellence (DFG-EXC 2046) through the project EP4-4.

	\bibliographystyle{plain}
	\bibliography{biblyo}

\begin{thebibliography}{10}

\bibitem{sahalia2007mle}
Y.~A{\"i}t-Sahalia and R.~Kimmel.
\newblock Maximum likelihood estimationof stochastic volatility models.
\newblock {\em Journal of Financial Economics}, 83(2):413--452, 2007.

\bibitem{avellaneda1995pricing}
Marco Avellaneda, Arnon Levy, and Antonio Par{\'a}s.
\newblock Pricing and hedging derivative securities in markets with uncertain
  volatilities.
\newblock {\em Applied Mathematical Finance}, 2(2):73--88, 1995.

\bibitem{beck2019machine}
Christian Beck, E~Weinan, and Arnulf Jentzen.
\newblock Machine learning approximation algorithms for high-dimensional fully
  nonlinear partial differential equations and second-order backward stochastic
  differential equations.
\newblock {\em Journal of Nonlinear Science}, 29(4):1563--1619, 2019.

\bibitem{black1973pricing}
Fischer Black and Myron Scholes.
\newblock The pricing of options and corporate liabilities.
\newblock {\em Journal of political economy}, 81(3):637--654, 1973.

\bibitem{cheridito2007second}
Patrick Cheridito, Halil~Mete Soner, Nizar Touzi, and Nicolas Victoir.
\newblock Second-order backward stochastic differential equations and fully
  nonlinear parabolic pdes.
\newblock {\em Communications on Pure and Applied Mathematics: A Journal Issued
  by the Courant Institute of Mathematical Sciences}, 60(7):1081--1110, 2007.

\bibitem{crandall2000lp}
Michael~G Crandall, Maciej Kocan, and A~{\'S}wiech.
\newblock Lp-theory for fully nonlinear uniformly parabolic equations:
  Parabolic equations.
\newblock {\em Communications in Partial Differential Equations},
  25(11-12):1997--2053, 2000.

\bibitem{da2016alpha}
Jos{\'e} Da~Fonseca and Claude Martini.
\newblock The $\alpha$-hypergeometric stochastic volatility model.
\newblock {\em Stochastic Processes and their Applications}, 126(5):1472--1502,
  2016.

\bibitem{fouque2018uncertain}
Jean-Pierre Fouque and Ning Ning.
\newblock Uncertain volatility models with stochastic bounds.
\newblock {\em SIAM Journal on Financial Mathematics}, 9(4):1175--1207, 2018.

\bibitem{fouque2003singular}
Jean-Pierre Fouque, George Papanicolaou, Ronnie Sircar, and Knut Solna.
\newblock Singular perturbations in option pricing.
\newblock {\em SIAM Journal on Applied Mathematics}, 63(5):1648--1665, 2003.

\bibitem{fouque2011multiscale}
Jean-Pierre Fouque, George Papanicolaou, Ronnie Sircar, and Knut S{\o}lna.
\newblock {\em Multiscale stochastic volatility for equity, interest rate, and
  credit derivatives}.
\newblock Cambridge University Press, 2011.

\bibitem{fouque2014approximation}
Jean-Pierre Fouque and Bin Ren.
\newblock Approximation for option prices under uncertain volatility.
\newblock {\em SIAM Journal on Financial Mathematics}, 5(1):360--383, 2014.

\bibitem{gatheral2006smile}
J.~Gatheral.
\newblock {\em The Volatility Surface: A Practitioner's Guide}.
\newblock Wiley Finance, 2006.

\bibitem{giga2010nonlinear}
Mi-Ho Giga, Yoshikazu Giga, and J{\''u}rgen Saal.
\newblock {\em Nonlinear partial differential equations: Asymptotic behavior of
  solutions and self-similar solutions}, volume~79.
\newblock Springer Science et Business Media, 2010.

\bibitem{gilli2021smile}
M.~Gilli and E.~Schumann.
\newblock Calibrating option pricing models with heuristics.
\newblock In A.~Brabazon, M.~O'Neill, and D.~Maringer, editors, {\em Natural
  Computing in Computational Finance}, volume 380 of {\em Studies in
  Computational Intelligence}, pages 9--37. Springer, Berlin, 2012.

\bibitem{guyon2013nonlinear}
Julien Guyon and Pierre Henry-Labordere.
\newblock {\em Nonlinear option pricing}.
\newblock CRC Press, 2013.

\bibitem{hartmann2019chaos}
Carsten Hartmann, Omar Kebiri, Lara Neureither, and Lorenz Richter.
\newblock Variational approach to rare event simulation using least-squares
  regression.
\newblock {\em Chaos}, 29(6):063107, 2019.

\bibitem{heston1993closed}
Steven~L Heston.
\newblock A closed-form solution for options with stochastic volatility with
  applications to bond and currency options.
\newblock {\em The review of financial studies}, 6(2):327--343, 1993.

\bibitem{hull1987pricing}
John Hull and Alan White.
\newblock The pricing of options on assets with stochastic volatilities.
\newblock {\em The journal of finance}, 42(2):281--300, 1987.

\bibitem{hurn2015mle}
A.~Hurn, K.~Lindsay, and A.~McClelland.
\newblock Estimating the parameters of stochastic volatility models using
  option price data.
\newblock {\em Journal of Business and Economic Statistics}, 33(4):579--594,
  2015.

\bibitem{lee2010detecting}
Suzanne~S Lee and Jan Hannig.
\newblock Detecting jumps from l{\'e}vy jump diffusion processes.
\newblock {\em Journal of Financial Economics}, 96(2):271--290, 2010.

\bibitem{lepage2006continuous}
Thomas Lepage, Stephan Lawi, Paul Tupper, and David Bryant.
\newblock Continuous and tractable models for the variation of evolutionary
  rates.
\newblock {\em Mathematical biosciences}, 199(2):216--233, 2006.

\bibitem{linos1983optimal}
Pierre-Louis Linos.
\newblock Optimal control of diffustion processes and hamilton-jacobi-bellman
  equations part i: the dynamic programming principle and application.
\newblock {\em Communications in partial differential equations},
  8(10):1101--1174, 1983.

\bibitem{lions1983optimal}
Pierre-Louis Lions.
\newblock Optimal control of diffusion processes and hamilton--jacobi--bellman
  equations part 2: viscosity solutions and uniqueness.
\newblock {\em Communications in partial differential equations},
  8(11):1229--1276, 1983.

\end{thebibliography}
	\it
	\noindent

\end{document}